\documentclass[11pt,a4paper]{article}

\usepackage{amsmath}
\usepackage{amsthm}
\usepackage{amsfonts}
\usepackage{amssymb}
\usepackage{graphicx}
\usepackage{bm}
\usepackage{cite}
\usepackage[justification=centering]{caption}
\usepackage{algorithm}
\usepackage{algorithmic}
\usepackage{authblk}

\title{An Accuracy-Assured Privacy-Preserving Recommender System for Internet Commerce}

\author{
Zhigang Lu,  Hong Shen\\
       The University of Adelaide\\
       \{zhigang.lu, hong.shen\}@adelaide.edu.au
}

\date{}

\newtheorem{lemma}{Lemma}[]
\newtheorem{theorem}{Theorem}[]
\newtheorem{definition}{Definition}[]

\begin{document}

\maketitle

\begin{abstract}
$Recommender\ systems$, tool for predicting users' potential preferences by computing history data and users' interests, show an increasing importance in various Internet applications such as online shopping. As a well-known recommendation method, neighbourhood-based collaborative filtering has attracted considerable attention recently. The risk of revealing users' private information during the process of filtering has attracted noticeable research interests. Among the current solutions, the probabilistic techniques have shown a powerful privacy preserving effect. The existing methods deploying probabilistic methods are in three categories, one\cite{MCSHERRY2009} adds differential privacy noises in the covariance matrix; one\cite{ADAMOPOULOS2014} introduces the randomisation in the neighbour selection process; the other\cite{ZHU2014} applies differential privacy in both the neighbour selection process and covariance matrix. When facing $k$ Nearest Neighbour ($k$NN) attack, all the existing methods provide no data utility guarantee, for the introduction of global randomness. In this paper, to overcome the problem of recommendation accuracy loss, we propose a novel approach, Partitioned Probabilistic Neighbour Selection, to ensure a required prediction accuracy while maintaining high security against $k$NN attack. We define the sum of $k$ neighbours' similarity as the accuracy metric $\alpha$, the number of user partitions, across which we select the $k$ neighbours, as the security metric $\beta$. We generalise the $k$ Nearest Neighbour attack to $\beta k$ Nearest Neighbours attack. Differing from the existing approach that selects neighbours across the entire candidate list randomly, our method selects neighbours from each exclusive partition of size $k$ with a decreasing probability. Theoretical and experimental analysis show that to provide an accuracy-assured recommendation, our Partitioned Probabilistic Neighbour Selection method yields a better trade-off between the recommendation accuracy and system security.

\vspace{6pt}\textbf{Keywords:} Privacy preserving, Differential privacy, Neighbourhood-based collaborative filtering recommender systems, Internet Commerce
\end{abstract}

\section{Introduction}
\label{INTRO}
$Recommender$ $systems$ predict users' potential preferences by aggregating history data and users' interests. Recently, an increasing importance of recommender systems has been shown in various Internet applications. For example, Amazon has been receiving benefits for a decade from the recommender systems by providing personal recommendation to their customers, and Netflix posted a one million U.S. dollars award for improving their recommender system to make their business more profitable \cite{SCHAFER1999, EKSTRAND2011, KABBUR2013}. Currently, in recommender systems, Collaborative Filtering (CF) is a famous technology with three main popular techniques \cite{LIU2011}, i.e., neighbourhood-based methods \cite{HERLOCKER2002}, association rules based prediction \cite{SARWAR2001}, and matrix factorisation \cite{KOREN2009}. Among these techniques, neighbourhood-based methods are the most widely used in the industry because of its easy implementation and high prediction accuracy.

One of the most popular neighbourhood-based method is $k$ Nearest Neighbour ($k$NN) which provides recommendations by aggregating the opinions of a user's $k$ nearest neighbours \cite{ADOMAVICIUS2005}. Although $k$NN recommender systems present very good performance of recommendation accuracy efficiently, the risk of revealing users' private information during the process of filtering is still a growing concern, e.g., the $k$NN attack presented by Calandrino et al. \cite{CALANDRINO2011} exploits the property that the users are more similar when sharing same rating on corresponding items to reveal user's private data. Thus presenting an efficient privacy preserving neighbourhood-based CF algorithm against $k$NN attack, which achieves a trade-off between the system security and recommendation accuracy, has been a natural research interest.

The literature in CF recommender systems has developed several approaches to preserve users' privacy. Generally, cryptographic, obfuscation, perturbation, probabilistic methods and differential privacy are applied \cite{ZHU2014}. Among them, cryptographic methods \cite{ERKIN2010, NIKOLAENKO2013} provide the most reliable security but the unnecessary computational cost cannot be ignored. Obfuscation methods \cite{PARAMESWARAN2007, WEINSBERG2012} and Perturbation methods \cite{BASU2012, BILGE2012} introduce designed random noise into the original matrix to preserve customers' sensitive information; however the magnitude of noise is hard to calibrate in these two types of methods \cite{DWORK2006B, ZHU2014}. The probabilistic methods \cite{ADAMOPOULOS2014} provided a similarity based weighted neighbour selection of the $k$ nearest neighbours. Similar to perturbation, McSherry et al. \cite{MCSHERRY2009} presented a naive differential privacy method which adds calibrated noise into the covariance (similarity between users/items) matrix. Similar to the probabilistic neighbour selection \cite{ADAMOPOULOS2014}, Zhu et al. \cite{ZHU2014} proposed a Private Neighbour Selection to preserve privacy against $k$NN attack by introducing differential privacy in selecting the $k$ nearest neighbours randomly (also adding noise into covariance matrix with differential privacy). Although the methods in \cite{MCSHERRY2009, ZHU2014, ADAMOPOULOS2014} successfully preserve users' privacy against $k$NN attack, the low prediction accuracy due to the global randomness should be noted. Even worse, \cite{ZHU2014} failed to maintain differential privacy in the process of neighbour selection. Therefore, none of the existing privacy preserving CF recommender systems can provide enough utility while preserving users' private information.

\textbf{Motivation.} The current privacy preserving neighbourhood-based CF methods did not guarantee the data utility against $k$NN attack. Therefore, in this paper, we aim to present a privacy preserving neighbourhood-based CF recommendation scheme which satisfies the following properties:
\begin{itemize}
\item[(1)] Easy implementation.
\item[(2)] Absolutely keep differential privacy.
\item[(3)] Significantly decrease the magnitude of noise in differential privacy.
\item[(4)] Quantify the level of recommendation accuracy and system security.
\end{itemize}

Actually, it is clear that the probabilistic methods (including naive probabilistic methods and differential privacy methods) are efficient methods against $k$NN attack; however, because of the global noises, the neighbour quality, namely the prediction accuracy, is impacted significantly. Thus, to decrease the magnitude of differential privacy noise, we may propose the following approach: we can simply add Laplace noise to the final rating prediction after the normal $k$NN CF recommendation. But Sarathy et al. has shown in \cite{SARATHY2011} that the above method will release users' privacy because Laplace mechanism does not work well in numeric data. So, to control the neighbour quality and to decrease the magnitude of noise, it is natural to avoid the global randomness and repeatedly adding noise. Therefore, we present a partitioned probabilistic neighbour selection method without any perturbations in the process of rating prediction.

\textbf{Contributions.} In this paper, to overcome the problems of low recommendation accuracy, we propose a novel method, Partitioned Probabilistic Neighbour Selection. The main contributions of this paper are:

(1) We expand the classic $k$NN attack to a more general case, $\beta$-$k$NN attack, which flexibly adjusts the size of fake user's set to improve the attack effectiveness. $\beta$ is essentially regarded as a security measure denoting the degree of difficulty for an attacker to break the neighbourhood-based CF recommender systems. We are the first to consider the case when $\beta>1$.

(2) To protect users' data privacy against $\beta$-$k$NN attack, we propose a novel differential privacy preserving neighbourhood-based CF method, which ensures a required prediction accuracy while achieving a better trade-off between the system security and recommendation accuracy against $k$NN attack.

(3) To the best of our knowledge, we are the first to propose a theoretical analysis of the recommendation accuracy and system security on the recommendation results from any randomised neighbour selection methods in the neighbourhood-based CF recommender systems. Previous related work only gave the experimental analysis on the same issues.

\textbf{Organisation.} The rest of this paper is organised as follows: In Section \ref{RW}, we summarise both the advantages and disadvantages in the existing privacy preserving methods on CF recommender systems. In Section \ref{PRE}, we introduce the relevant background knowledge in this paper. In Section \ref{ATK}, we introduce an existing attack to neighbourhood-based CF recommender systems, then expand it to a general case, $\beta$-$k$NN attack. Next, We proposed a novel differential privacy recommendation approach, Partitioned Probabilistic Neighbour Selection, in Section \ref{PPNS}. Afterwards, the theoretical analysis of our approach on the performance of both recommendation accuracy and system security are provided in Section \ref{EVA}. Then, in Section \ref{EXP}, we show the experimental evaluation results. Finally, in Section \ref{CON}, we conclude this paper.

\section{Related Work}
\label{RW}
A noticeable number of literature has been published to preserve customers' private data in recommender systems. However, Calandrino et al. \cite{CALANDRINO2011} proposed a neighbourhood-based CF attack, $k$NN attack, which is a serious privacy threat to the neighbourhood-based CF recommender systems in e-commerce, e.g., Amazon. In this section, we briefly discuss some of the research literature in privacy preserving neighbourhood-based CF recommender systems.

\subsection{Traditional Privacy Preserving Recommender Systems}
Amount of traditional privacy preserving methods have been developed in CF recommender systems \cite{ZHU2014}, including cryptographic \cite{ERKIN2010, NIKOLAENKO2013}, obfuscation \cite{PARAMESWARAN2007, WEINSBERG2012}, perturbation \cite{BASU2012, BILGE2012} and probabilistic methods \cite{ADAMOPOULOS2014}. Erkin et al. \cite{ERKIN2010} applied homomorphic encryption and secure multi-party computation in privacy preserving recommender systems, which allows users to jointly compute their data to receive recommendation without sharing the true data with other parties. Nikolaenko et al. \cite{NIKOLAENKO2013} combined a famous recommendation technique, matrix factorization, and a cryptographic method, garbled circuits, to provide recommendations without learning the real user ratings in database. The Cryptographic methods provide the highest guarantee for both prediction security ans system security by introducing encryption rather than adding noise to the original record. Unfortunately, unnecessary computational cost impacts its application in industry \cite{ZHU2014}. Obfuscation and perturbation are two similar data processing methods. In particular, obfuscation methods aggregate a number of random noises with real users rating to preserve user's sensitive information. Parameswaran et al. \cite{PARAMESWARAN2007} proposed an obfuscation framework which exchanges the sets of similar items before submitting the user data to CF server. Weinsberg et al. \cite{WEINSBERG2012} introduced extra reasonable ratings into user's profile against inferring user's sensitive information. Perturbation methods modify the user's original ratings by a selected probability distribution before using these ratings. Particularly, Bilge et al. \cite{BILGE2012} added uniform distribution noise to the real ratings before the utilisation of user's rating in prediction process. While, Basu et al. \cite{BASU2012} regarded the deviation between two items as the adding noise. Both perturbation and obfuscation obtain good trade-off between prediction accuracy and system security due to the tiny data perturbation, but the magnitude of noise or the percentage of replaced ratings are not easy to be calibrated \cite{DWORK2006B, ZHU2014}. The probabilistic method \cite{ADAMOPOULOS2014} applied weighted sampling in neighbour selection which preserves users' privacy against $k$NN successfully; however, it cannot provide enough accuracy due to its global randomness. Because the performance of the neighbourhood-based CF methods largely depends on the quality of neighbours. We suppose the top $k$ neighbour as the highest quality neighbour set, the randomised weighted selection process will return neighbours with lower similarity with a high probability. Then the prediction accuracy will be impacted significantly \cite{ZHU2014}. Therefore, achieving a trade-off between privacy and utility, while calibrating the adding noise are difficult tasks for these techniques.

\subsection{Differential Privacy Recommender Systems}
\label{ZHU}
As a well-known privacy definition, the differential privacy technology \cite{DWORK2006} has been applied in the research of privacy preserving recommender systems. For example, McSherry et al. \cite{MCSHERRY2009} provided the first differential privacy neighbourhood-based CF recommendation algorithm. In fact, their naive differential privacy protects the neighbourhood-based CF recommender systems against $k$NN attack successfully, as they added Laplace noise into the covariance (similarity between users/items) matrix globally, so that the output $k$ nearest neighbours set is no longer the original top $k$ neighbours. However, the global noise decreases the accuracy of their recommendation algorithms significantly.

Another differential privacy neighbourhood-based CF recommender systems algorithm is proposed by Zhu et al. \cite{ZHU2014} which inspired this study. It aims to provide better prediction accuracy than McSherry et al. \cite{MCSHERRY2009} while aiming to keep differential privacy at both neighbour selection stage and rating prediction stage. They proposed a Private Neighbour Collaborative Filtering (PNCF) by introducing exponential differential privacy \cite{MCSHERRY2007} to the process of neighbour selection to guarantee the system security against $k$NN attack. After selecting the $k$ neighbours, same with McSherry et al. \cite{MCSHERRY2009}, they also added Laplace noise into the similarity matrix to make the final prediction.

Unlike the $k$ nearest neighbour method which selects the $k$ most similar candidates, the PNCF method \cite{ZHU2014} randomly selects the $k$ neighbours with each candidate $u_i$'s weight $\omega_i$. According to exponential mechanism of differential privacy, the selection weight is measured by a score function and its corresponding sensitivity as follow,
\begin{equation}
\label{OMEGA}
\omega_i=\exp(\frac{\epsilon}{4k\times RS}q_a(U(u_a),u_i)),
\end{equation}
where $q$ is the score function, $RS$ is the Recommendation-Aware Sensitivity of score function $q$ for any user pairs $u_i$ and $u_j$, $\epsilon$ is differential privacy parameter, and $U(u_a)$ is the set of user $u_a$'s candidate list. For a user $u_a$, the score function $q$ and its Recommendation-Aware Sensitivity are defined as follows:
\begin{equation}
\label{SelectionFunc}
q_a(U(u_a),u_i)=sim_{ai},
\end{equation}
\begin{equation}
\label{RS}
\begin{array}{c}
RS=\max\left\{ \underset{s\in S_{ij}}{\max}\left(\frac{r_{is}\cdot r_{js}}{\left\Vert r_{i}'\right\Vert \left\Vert r_{j}'\right\Vert }\right),\right.\left.\underset{s\in S_{ij}}{\max}\left(\frac{r_{is}\cdot r_{js}\left(\left\Vert r_{i}\right\Vert \left\Vert r_{j}\right\Vert -\left\Vert r_{i}'\right\Vert \left\Vert r_{j}'\right\Vert \right)}{\left\Vert r_{i}\right\Vert \left\Vert r_{j}\right\Vert \left\Vert r_{i}'\right\Vert \left\Vert r_{j}'\right\Vert }\right)\right\},
\end{array}
\end{equation}
where $r_{is}$ is user $u_i$'s rating on item $t_s$, $sim_{ai}$ is the similarity between user $u_a$ and $u_i$, $r_i$ is user $u_i$'s average rating on every item, $S_{ij}$ is the set of all items co-rated by both users $i$ and
$j$, i.e., $S_{ij}=\{s\in S|r_{is}\neq \varnothing \ \&\
r_{js}\neq \varnothing \}$.

However, the above naive differential privacy neighbour selection is nearly the same to the probabilistic neighbour selection \cite{ADAMOPOULOS2014}. To address the above problem of low prediction accuracy in \cite{ADAMOPOULOS2014}, a truncated parameter $\lambda$ was introduced in \cite{ZHU2014}. Simply speaking, the candidates whose similarity is greater than $(sim(a,k)+\lambda)$ are selected to the neighbour set, while, whose similarity is less than $(sim(a,k)-\lambda)$ will not be selected, where $sim(a,k)$ denotes the similarity of user $u_a$'s $k$th neighbour. Theorem 3.1 in \cite{ZHU2014} provided an equation to calculate the value of $\lambda$, i.e. $\lambda=\min(sim(a,k),\frac{4k\cdot RS}{\epsilon}\ln\frac{k(n-k)}{\rho})$, where $\rho$ is a constant, $0<\rho<1$.

However, we observe that the above idea in \cite{ZHU2014} has three weaknesses. Firstly, it adds random noise in the process of neighbour selection twice; however, it is not necessary. Because we can preserve privacy against $k$NN successfully only by introducing randomness once, the extra randomness will decrease the prediction accuracy significantly. Secondly, the value of $\lambda$ may not be achievable. This is because when computing the value of $\lambda$ by $\rho$, it results in a good theoretical recommendation accuracy, but does not yield a good experimental recommendation accuracy on the given test datasets in \cite{ZHU2014}. So the PNCF method \cite{ZHU2014} will actually be a method of Global Probabilistic Neighbour Selection \cite{ADAMOPOULOS2014} and cannot guarantee any recommendation accuracy. Thirdly, the PNCF scheme breaks differential privacy in the process of neighbour selection. Suppose there is a tiny change in the dataset, then the value of similarity between target user $u_a$ and other users $u_i$ in the candidate list will change. There may exist a user $u_c$ whose probability of being selected may change from 0 to $x>0$, then the ratio between the two probabilities will be 0 or infinite, none of which satisfy Definition \ref{DIFFPRI} in Section \ref{DIFF}.

\section{Preliminaries}
\label{PRE}
In this section, we introduce the foundational concepts and mathematical model related with this paper in collaborative filtering, differential privacy, and Wallenius' non-central hyper-geometric distribution.
\subsection{$k$ Nearest Neighbour Collaborative Filtering}
A collaborative filtering based recommender system predicts users' potential preferences by aggregating the relevant historical data. Collaborative filtering, a popular technique in recommender systems, is in three categories: neighbourhood-based methods, association rules based methods, and matrix factorisation methods \cite{LIU2011}. The neighbourhood-based methods generally provides recommendations by combining the opinions of a user's $k$ nearest neighbours \cite{ADOMAVICIUS2005}.

Neighbour Selection and Rating Prediction are two main stages in neighbourhood-based CF \cite{ZHU2014}. At the Neighbour Selection stage, a target user $u_a$'s neighbours are selected according to their similarity value in the target user $u_a$'s similarity array $\mathcal{S}_a$, where similarities between any two users/items are calculated by a measurement metric. Two of the most popular similarity measurement metrics are the Pearson correlation coefficient and Cosine-based Similarity \cite{ADOMAVICIUS2005}. In $k$NN method, we select the $k$ most similar neighbours of a target user/item.
\begin{itemize}
\item[(1)]
Pearson Correlation Coefficient (user-based):
\begin{equation}
sim_{ij}=\frac{\sum_{s\in S_{ij}}(r_{is}-\bar
r_i)(r_{js}-\bar r_j)}{\sqrt{\sum_{s\in
S_{ij}}(r_{is}-\bar r_i)^2\sum_{s\in
S_{ij}}(r_{js}-\bar r_j)^2}},
\end{equation}
\item[(2)]
Cosin-based Similarity (user-based):
\begin{equation}
\label{COS}
\begin{array}{ll}
sim_{ij}&= \cos(\bm{r}_i,\bm{r}_j)={\frac{\bm{r}_i\cdot
\bm{r}_j}{\|\bm{r}_i\|\times\|\bm{r}_j\|}}\\
&= \frac{\sum_{s\in S_{ij}}r_{is}r_{js}}{\sqrt{\sum_{s\in S_{ij}}r_{is}^2}\sqrt{\sum_{s\in S_{ij}}r_{js}^2}},
\end{array}
\end{equation}
\end{itemize}
where $r_{is}$ is user $u_i$'s rating on item $t_s$, $r_{is}\in\mathcal{R}$, $\mathcal{R}$ is the user-item rating dataset, $sim_{ij}$ is the similarity between user $u_i$ and user $u_j$, $\bar{r}_i$ is user $u_i$'s average rating on every item, $S_{ij}$ is the set of all items co-rated by both users $i$ and $j$, i.e., $S_{ij}=\{s\in S|r_{is}\neq \varnothing \ \&\ r_{js}\neq \varnothing \}$.

At the stage of Rating Prediction in user-based CF methods, the predicted rating $\hat{r}_{ax}$ of user $u_a$ on item $t_x$ is calculated as an aggregation of other users' rating on item $t_x$ \cite{ADOMAVICIUS2005, ZHU2014}. The prediction of $\hat{r}_{ax}$ is computed as follow:
\begin{equation}
\label{PRED}
\hat{r}_{ax}=\frac{\sum_{u_i \in N_k(u_a)}sim(a,i)r_{ix}}{\sum_{u_i \in N_k(u_a)}|sim(a,i)|},
\end{equation}
where, $N_k(u_a)$ is a sorted set which contains user $u_a$'s $k$ nearest neighbours, $N_k(u_a)$ is sorted by similarity in a descending order, $sim(a,i)$ is the $i$th neighbour of $u_a$ in $N_k(u_a)$.

\subsection{Differential Privacy}
\label{DIFF}
Informally, differential privacy \cite{DWORK2006, DWORK2008} is a scheme that minimises the sensitivity of output for a given statistical operation on two different (differentiated in one record to protect) datasets. Specifically, differential privacy guarantees no matter whether one specific record appears in a database, the privacy mechanism will shield the specific record to the adversary. The strategy of differential privacy is adding a random noise to the result of a query function on the database. 

To understand the spirit of differential privacy clearly, several items will be introduced in advance. Firstly, $X(x_1, x_2, \cdots, x_n)$ and $X^{\prime}(x_{1}^{\prime}, x_{2}^{\prime}, \cdots, x_{n}^{\prime})$ are two databases with $n$ entries which differ in only one entry, where $x_i$ and $x_{1}^{\prime}$ are the $i$th entry of $X$ and $X^{\prime}$. We call $X$ and $X^{\prime}$ are neighbouring dataset. Secondly, $f(X)$ is the query function on database $X$, the respond is the combination of the real answer $a = f(X)$ and a chosen random noise. Thirdly, the privacy mechanism $\mathcal{T}$, namely, the respond, is computed by $\mathcal{T}(X) = f(X) + Noise$. A formal definition of Differential Privacy is shown as follow:

\begin{definition}[$\epsilon$-Differential Privacy \cite{DWORK2006}]
\label{DIFFPRI}
A randomised mechanism $\mathcal{T}$ is $\epsilon$-differential privacy if for all neighbouring datasets $X$ and $X^{\prime}$, and for all outcome sets $S \subseteq Range(T)$, $\mathcal{T}$ satisfies: $\Pr[\mathcal{T}(X) \in S] \leq \exp(\epsilon)\cdot\Pr[\mathcal{T}(X^{\prime}) \in S]$, where $\epsilon$ is a privacy budget.
\end{definition}
The privacy budget $\epsilon$ is set by the database owner. Usually, a smaller $\epsilon$ denotes a higher privacy guarantee because the privacy budget $\epsilon$ reflects the magnitude of difference between two neighbouring datasets.

There are two main applications of the randomised mechanism $\mathcal{T}$: the Laplace mechanism \cite{DWORK2006} and the Exponential mechanism \cite{MCSHERRY2007}. The definitions are shown as below:
\begin{definition}[Laplace Mechanism \cite{DWORK2006}]
Let a query function $f$: $\mathbb{R}\to\mathbb{R}^d$, the $\epsilon$-differential privacy mechanism $\mathcal{T}$ obeys that $\mathcal{T}(X)=f(X)+Lap^{-1}(\frac{\Delta f}{\epsilon}, \Pr)^d$, where the sensitivity of function $f$, $\Delta f = \max{|f(X) - f(X^{\prime})|}$, for all neighbouring datasets $X$, $X^{\prime}$ $\in \mathcal{D}^n$, and $d$ represents the dimension.
\end{definition}
\begin{definition}[Exponential Mechanism \cite{MCSHERRY2007}]
Given a score function of a database $X$, $q(X,x)$, which reflects the score of query respond $x$. The exponential mechanism $\mathcal{T}$ provides $\epsilon$-differential privacy, if $\mathcal{T}(X)$ = \{the probability of a query respond $x$ $\propto$ $\exp({\frac{\epsilon\cdot q(X,x)}{2\Delta q}})$\}, where $\Delta q=\max|q(X,x)-q(X^{\prime},x)|$, denotes the sensitivity of $q$.
\end{definition}

\subsection{Wallenius' Non-central Hyper-geometric Distribution}
\label{WNHD}
Wallenius' non-central hyper-geometric distribution is a distribution of weighted sampling without replacement. Formally, it is defined as follow \cite{FOG2008B}: We assume there are $c$ distinct categories in the population, each category contains $m_i$ individuals, $i \in [1, c]$. All the individuals from category $i$ have the same weight $\omega_i$, $i \in [1, c]$. The probability of an individual is sampled at a given draw is proportional to its weight $\omega_i$. Let $\bm{x}_v=(x_{1v},x_{2v},\ldots,x_{cv})$ denote the total number of the individuals in each colour sampled after the first $v$ draws. The probability that the next draw gives a individual of colour $i$ is:
\begin{equation}
p_{i(v+1)}(\bm{x}_v)=\frac{(m_i-x_{iv})\omega_i}{\sum_{j=1}^c(m_j-x_{jv})\omega_j}.
\end{equation}
The weighted sampling process without replacement is repeatedly until $k$ individuals have been retained, namely, $k = \sum_{i = 1}^c{x_i}$, where $x_i$ denotes the number of individuals sampled from category $i$ by Wallenius' non-central hypergeometric distribution.

Wallenius \cite{WALLENIUS1963} proposed the probability mass function for this distribution in the univariate case $(c = 2)$. Chesson \cite{CHESSON1976} expanded Wallenius's solution to the multivariate case $(c > 2)$. In this paper, we focus on the multivariate Wallenius' non-central hyper-geometric distribution's probability mass function because we regard one user/item in a recommender system as one individual in Wallenius' non-central hyper-geometric distribution. The multivariate probability mass function (PMF) is shown as blow:
\begin{equation}
\label{PMF}
mwnchypg=\bm{\Lambda(x)I(x)},
\end{equation}
where
$\bm{\Lambda(x)}=\prod_{i=1}^c
\begin{pmatrix}m_{i}\\
x_{i}
\end{pmatrix}, \bm{I(x)}=\int_0^1\prod_{i=1}^c(1-t^{\omega_i/d})^{x_i}\text{d}t, d=\bm{\omega\cdot(m-x)}=\sum_{i=1}^c\omega_i(m_i-x_i),$ $\bm{x}=(x_1,x_2,\ldots,x_c)$, $\bm{m}=(m_1,m_2,\ldots,m_c)$, $\bm{\omega}=(\omega_1,\omega_2,\ldots,\omega_c)$.

While in this paper, we mainly use the following properties to evaluate different probabilistic relevant approaches. Manly \cite{MANLY1974} gave the approximated solution $\bm{\mu}^*=(\mu_{1}^*,\mu_{2}^*,\ldots,\\\mu_{c}^*)$ to the mean $\bm{\mu}=(\mu_{1},\mu_{2},\ldots,\mu_{c})$ of $\bm{x}$ after the final draw:
\begin{equation}
\label{approx}
\left(1-\frac{\mu_1^*}{m_1}\right)^{1/\omega_1}=\left(1-\frac{\mu_2^*}{m_2}\right)^{1/\omega_2}=\ldots=\left(1-\frac{\mu_c^*}{m_c}\right)^{1/\omega_c},
\end{equation}
where $\sum_{i=1}^c\mu_i^*=k$, $\forall i\in C: 0\leq \mu_i^* \leq m_i.$

Fog \cite{FOG2008} stated the following properties of Equation \eqref{approx}: firstly, the solution $\bm{\mu}^*$ is valid under the conditions that $\forall i \in C: m_i>0$ and $\omega_i>0$. Secondly, the mean given by Equation \eqref{approx} is a good approximation in most cases. Thirdly, Equation \eqref{approx} is exact when all $\omega_i$ are equal.

\section{A Generalised Privacy Attack for Recommender Systems}
\label{ATK}
In this section, we firstly introduce a popular attack, $k$ nearest neighbour attack, then we expand the concept to a general attack, $\beta$-$k$ nearest neighbour attack.
\subsection{$k$ Nearest Neighbour Attack}
Calandrino et al. \cite{CALANDRINO2011} stated a user-based attack called \textit{k Nearest Neighbour} ($k$NN) attack. Simply, the $k$NN attack exploits the property that the users are more similar when sharing same rating on corresponding items to reveal user's private data.

We assume that the recommendation algorithm ($k$NN CF recommendation) and its parameter $k$ are known to the attacker. Furthermore, the attacker's auxiliary information consists of a target user $u_a$'s partial history rating values, i.e., he already knows the ratings of $m$ items that $u_a$ has rated. Usually, $m \approx 8$. He aims to catch $u_a$'s transactions that he does not yet know about.

To achieve this goal, the attacker firstly creates $k$ fake users who have the same ratings with $u_a$ only on the $m$ items. With a high probability, each fake user's $k$ nearest neighbours set $N_k(\text{fake user})$ will include the other $k-1$ fake users and the target user $u_a$. Because the target user $u_a$ is the only neighbour who has ratings on the items which are not rated by the fake users, to provide recommendations on these items to the fake users, the recommender system has to give $u_a$'s rating to the fake users directly. Obviously, the fake users learn the target user $u_a$'s whole rating list successfully with $k$NN attack.

\subsection{$\beta$-$k$ Nearest Neighbours Attack}
According to the existing privacy preserving neighbourhood-based CF recommendation methods, we expand the $k$NN attack to a more general case, named $\beta$-$k$ Nearest Neighbour ($\beta$-$k$NN) attack.

As we know, to preserve the target user $u_a$'s private information against $k$NN attack, we should avoid selecting the true $k$ nearest neighbours, so the existing methods applied the randomness techniques. However, suppose the final $k$ neighbours are selected from the top $\beta k$ users of $u_a$'s candidate list, also the parameters $\beta$ and $k$ are known to the attacker, the attacker would catch $u_a$'s private data with a high probability by creating $\beta k$ fake users. When $\beta$ is not great enough, it is still not difficult to break the privacy preserving neighbourhood-based CF recommender systems. Therefore, the $\beta$-$k$NN attack can flexibly adjust the size of fake user's set to improve the attack effectiveness. Actually, $k$NN attack can be regarded as 1-$k$NN attack in the expanded case of $\beta$-$k$NN attack.

In $\beta$-$k$NN attack, $\beta$ can be treated because a security measure as a greater value of $\beta$ represents a higher fraud cost. We will show the relationship between the prediction utility and $\beta$ in Section \ref{EVA}.

\section{Privacy Preservation by Partitioned Probabilistic Neighbour Selection}
\label{PPNS}
In this section, we firstly provide two performance metrics on the privacy preserving neighbourhood-based CF recommender systems against $\beta$-$k$NN attack. Then we propose our Partitioned Probabilistic Neighbour Selection algorithm based on the previous analysis.
\subsection{Performance Metrics}
\subsubsection{Accuracy Metric}
For any privacy preserving neighbourhood-based CF recommender systems, if the sum of similarity of the selected $k$ neighbours is greater, the predicted rating value will be better. The reason is simple: the neighbour is closer to the target user $u_a$ means the predicted result is more reliable, namely, we prefer the method which selects the greater similarity sum. Therefore, we define the accuracy metric $\alpha$ as the sum of $k$ selected neighbours' similarity.

Because we propose a random neighbour selection method, the accuracy metric $\alpha$ should be regarded as the expected sum of $k$ selected neighbours' similarity. However, it is not obvious to directly compute the expectation of the $k$ neighbours similarity sum: $\mathbb{E}(\sum_{i \in N_k(u_a)}sim(a,i))$, as we need to find all the user combinations and corresponding probabilities. So we give another way to compute this expectation,
\begin{equation}
\label{EXPSIMSUM}
\mathbb{E}(\sum_{i \in N_k(u_a)}sim(a,i))=\sum_{i=1}^n(sim(a,i)\mathbb{E}(x_i))=\sum_{i=1}^n{sim(a,i)\mu_i},
\end{equation}
see Section \ref{WNHD} for the definition of $x_i$ and $\mu_i$. So we compute the accuracy by the following equation in this paper:
\begin{equation}
\label{ALPHA}
\alpha=\sum_{i=1}^n{sim(a,i)\mu_i}.
\end{equation}

\subsubsection{Security Metric}
According to the $\beta$-$k$NN attack, suppose the final $k$ neighbours are selected from the top $\beta k$ users of $u_a$'s candidate list. We assume that the parameters $\beta$ and $k$ are known to the attacker, so the attacker would catch $u_a$'s privacy with a high probability through the same process of $k$NN attack by creating $\beta k$ fake users. When $\beta$'s value is not great, it is still not difficult to break the privacy preserving recommender systems. Therefore, we define $\beta$ as the security metric, the greater value of $\beta$ denotes the higher fraud cost for the attacker, namely, we want to achieve a trade-off between the security metric $\beta$ and a fixed prediction accuracy metric $\alpha$.

\subsection{Partitioned Probabilistic Neighbour Selection Algorithm}
\label{MODEL}
According to the motivation and previous analysis, we provide an original version of our Partitioned Probabilistic Neighbour Selection algorithm. We firstly partition the a target user's candidate list (descending order of similarity value) by the given $k$, then apply a geometric distribution on the candidate list to select $\lceil p(1-p)^{i-1}k\rceil$ neighbours (apply exponential differential privacy in every partition) from partition $i$ until we have a total of $k$ neighbours, where integer $i\in [1,+\infty)$, $p$ is a geometric distribution parameter. It is clear that our original partitioned probabilistic neighbour scheme satisfies property (1) (easy implementation) in Section \ref{INTRO}, for it does not introduce any extra computational cost. In fact, it is natural to regard the low neighbour quality as the noise in the process of neighbour selection, since the low neighbour quality has the same impact on the prediction accuracy as the noise. So our method satisfies property (3) (decreasing the magnitude of noise) in Section \ref{INTRO} in two ways: 1. it only adds noise in the process of neighbour selection. 2. it controls the neighbour quality by tuning the geometric distribution parameter $p$ in the process of neighbour selection. However, the original version does not satisfy the property (2) (keeping differential privacy) and (4) (quantifying the accuracy and security), we now show the reasons and modify it to satisfy the property (2) and (4).

In the original version, we select $\lceil p(1-p)^{i-1}k\rceil$ neighbours with exponential differential privacy from partition $i$ until we have $k$ neighbours. Actually, it breaks differential privacy with the same reason (see details in Section \ref{ZHU}) of the PNCF method \cite{ZHU2014}. Simply speaking, there may exist some users whose probability of selection will be changed from zero to a positive number because of a tiny change in rating set. To guarantee the prediction accuracy, we only modify the original scheme by changing the way we select the last neighbour (see details in next paragraph). The modified scheme keeps absolute differential privacy because no matter how we change the dataset, every candidate's probability of selection cannot be zero. To quantify the level of recommendation accuracy and system security, we use the performance metrics $\alpha$ and $\beta$. We compute the parameter $p$ and the security metric $\beta$ by a given $\alpha$ by Equation \eqref{final}.

Algorithm \ref{ALG} shows the Partitioned Probabilistic Neighbour Selection (PPNS) algorithm. In lines 1 to 5, we compute the necessary parameters by Equation \eqref{COS}, \eqref{RS}, \eqref{SelectionFunc}, \eqref{OMEGA} and \eqref{final}. In lines 6 to 18, we select the $k$ neighbours by Partitioned Probabilistic Neighbour Selection, then return the target user's $k$ neighbours and the security metric value $\beta$. We firstly mark all of the partitions as unvisited. Next, we select $\lceil p(1-p)^{i-1}k\rceil$ neighbours with exponential differential privacy from partition $i$ (mark this partition as visited) until we have a total of $k-1$ neighbours. Finally, we select the last neighbour from all the unvisited partitions. 
\begin{algorithm}[htb] 
\caption{Partitioned Probabilistic Neighbour Selection.} 
\label{ALG}
\begin{algorithmic}[1] 
\REQUIRE ~~\\
Original user-item rating set, $\mathcal{R}$;\\
Target user, $u_a$ and prediction item, $t_x$;\\
Number of neighbours, $k$;\\
Differential privacy parameter, $\epsilon$;\\
Accuracy metric, $\alpha$.
\ENSURE ~~\\
Target user $u_a$'s $k$-neighbour set, $N_{k}(u_a)$;\\
Security metric, $\beta$.
\STATE Compute the similarity list for target user $u_a$, $\mathcal{S}_a$;
\STATE Sort $\mathcal{S}_a$ in descending order, $\mathcal{S}_a$;
\STATE Compute exponential differential privacy sensitivity, $RS$;
\STATE Compute user $u_i$'s selection weight, $\omega_i$;
\STATE Compute the geometric distribution parameter, $p$;
\STATE Partition the sorted $S_a$ by $k$;
\FOR{$i=1$ to $n$}
\IF{Neighbour Number $\neq k-1$}
\STATE Select $\lceil p(1-p)^{i-1}k\rceil$ neighbours from partition $i$ to $N_{k}(u_a)$;
\STATE Mark partition $i$ as visited;
\STATE Neighbour Number += $\lceil p(1-p)^{i-1}k\rceil$;
\ELSE
\STATE break;
\ENDIF
\ENDFOR
\STATE Select one neighbour from unvisited partitions;
\STATE $\beta=$ last neighbour's partition index number;
\RETURN $N_{k}(u_a)$, $\beta$; 
\end{algorithmic}
\end{algorithm}

\section{Theoretical Analysis}
\label{EVA}
In this section, we use multivariate Wallenius' non-central hyper-geometric distribution to analyse any randomised neighbour selection methods on both performance of accuracy and security against $k$NN attack theoretically. The reason is both multivariate Wallenius' non-central hyper-geometric distribution and randomised neighbour selection methods are weighted sampling without replacement, the samples are selected one by one from universe, and the sampling weight is only depends on each sample's attribute, i.e., the ball's colour or user's similarity.

\subsection{Accuracy Analysis}
In this part, to analyse the accuracy performance, we will firstly modify the Equation \eqref{approx} to match with a general randomised neighbour selection method. As the selection weight in a general probabilistic neighbour selection method only relies on the user's similarity, we regard user $u_i$'s similarity $sim(a,i)$ as the sample's colour in multivariate Wallenius' non-central hyper-geometric distribution. Thus in randomised neighbour selection methods, $m_i=1$, $c=n$, $N=\sum_{i=1}^c{m_i}=\sum_{i=1}^n{m_i}=n$. Therefore, we rewrite the Equation \eqref{approx} as:
\begin{equation}
\label{MAIN}
A=(1-\mu_1)^{1/\omega_1}=(1-\mu_2)^{1/\omega_2}=
\ldots=(1-\mu_n)^{1/\omega_n},
\end{equation}
where $A$ is a constant.

Now we start evaluating the Partitioned Probabilistic Neighbour Selection by Equation \eqref{MAIN}. To make it easy, we also partition the candidate list in PNCF method \cite{ZHU2014} and Probabilistic Neighbour Selection \cite{ADAMOPOULOS2014} by the given $k$.

\begin{lemma}
\label{l1}
$\mathcal{C}$ is an $n$ sized set. We independently sample several samples with multivariate Wallenius' non-central hyper-geometric distribution from $\mathcal{C}$ twice, suppose $\mu_i$ and $\hat{\mu_i}$ are the expected number of sample $i$ from the two samplings. Then $\forall i\in [1, n]$, $\mu_i>\hat{\mu_i}\Leftrightarrow\sum_{i=1}^{n}\mu_i>\sum_{i=1}^{n}\hat{\mu_i}$.
\end{lemma}

\begin{proof}
Let $\sum_{i=1}^{n}\mu_i=X$, $\sum_{i=1}^{n}\hat{\mu_i}=\hat{X}$, $A=(1-\mu_i)^{1/{\omega_i}}$, $\hat{A}=(1-\hat{\mu_i})^{1/{\omega_i}}$.

(1) Proof of sufficient condition, $\mu_i>\hat{\mu_i}\Rightarrow \sum_{i=1}^{n}\mu_i>\sum_{i=1}^{n}\hat{\mu_i}$:

$\because$ the size of the set $\mathcal{C}$ keep the same.

$\therefore$ $\forall i\in [1, n], \mu_i>\hat{\mu_i} \Rightarrow \sum_{i=1}^{n}\mu_i>\sum_{i=1}^{n}\hat{\mu_i}.$

(2) Proof of Necessary condition, $\sum_{i=1}^{n}\mu_i>\sum_{i=1}^{n}\hat{\mu_i}\Rightarrow \mu_i>\hat{\mu_i}$:

According to Equation \eqref{MAIN}, we have,

$\begin{array}{ll}
 & A=(1-\mu_{i})^{1/{\omega_{1}}}\Rightarrow\mu_{i}=1-A^{1/{\omega_{1}}}\\
\Rightarrow & \sum_{i=1}^{n}\mu_{i}=k-\sum_{i=1}^{n}A^{1/{\omega_{1}}}=X.
\end{array}$

Similarly, $k-\sum_{i=1}^{n}\hat{A}^{1/{\omega_{1}}}=\hat{X}$.

$\because X=\sum_{i=1}^{n}\mu_i>\sum_{i=1}^{n}\hat{\mu_i}=\hat{X}$, and $\mu_i$ and $\hat{\mu}_i$ share the same $\omega_i$,

$
\begin{array}{lll}
\therefore & & \sum_{i=1}^{n}A^{1/{\omega_{1}}}<\sum_{i=1}^{n}\hat{A}^{1/{\omega_1}}\\
 & \Rightarrow & A<\hat{A}\\
& \Rightarrow & (1-\mu_{i})^{1/{\omega_1}}<(1-\hat{\mu}_{i})^{1/{\omega_1}}\\
 & \Rightarrow & \mu_i>\hat{\mu_i}.
\end{array}
$

Therefore, we have $\forall i\in [1, n]$, $\mu_i>\hat{\mu_i}\Leftrightarrow\sum_{i=1}^{n}\mu_i>\sum_{i=1}^{n}\hat{\mu_i}$.
\end{proof}

Lemma \ref{l1} shows the fact that when selecting neighbours with multivariate Wallenius' non-central hyper-geometric distribution by several randomised neighbour selection methods from a same sized partition, if one method selects more neighbours, then the expected number of each neighbour in that method is greater too, and vice versa.

\begin{lemma}
\label{l2}
The method, which selects more users from the first partition (contains user $u_1$ to $u_k$) of a descending order similarity list, yields a better rating prediction, i.e., $\sum_{i=1}^k\mu_i > \sum_{i=1}^k\hat{\mu_i}\Rightarrow \alpha \geq \hat{\alpha}$, where $\alpha$ denotes the accuracy metric value.
\end{lemma}

\begin{proof}
Let $X_j=\sum_{i\in{group_j}}\mu_i$,
$\hat{X}_j=\sum_{i\in{group_j}}\hat{\mu}_i$, e.g., $X_1=\sum_{i\in{group_1}}\mu_i=\sum_{i=1}^k{\mu_i}$. Assume an extreme case:

$\begin{array}{ccc}
X_{1} & > & \hat{X_{1}}\\
X_{2} & < & \hat{X_{2}}\\
X_{3} & < & \hat{X_{3}}\\
\vdots & < & \vdots
\end{array}
$

$\because$ $k=\sum_{j}{X_j}=\sum_{j}{\hat{X_j}},$

$\therefore$
$X_1-\hat{X_1}=(\hat{X_2}-X_2)+(\hat{X_3}-X_3)+\cdots$.

It is obvious that, in both sides of the above equation, every item $>0$. According to Lemma \ref{l1},
$X>\hat{X}\Leftrightarrow{\mu_i>\hat{\mu_i}}$, we have,

$\sum_{group_1}(\mu_i-\hat{\mu_i})=\sum_{group_2}(\hat{\mu_i}-\mu_i)+\sum_{group_3}(\hat{\mu_i}-\mu_i)+\cdots$,
and every $(\cdot)>0$.

$\because$ $1\geq sim(a,i)\geq sim(a,j)\geq 0$, $(i<j)$,

$
\begin{array}{lll}
\therefore\sum_{group_{1}}sim(a,i)(\mu_{i}-\hat{\mu_{i}}) & \geq & \sum_{group_{2}}sim(a,i)(\hat{\mu_{i}}-\mu_{i})\\
 &  & +\sum_{group_{3}}sim(a,i)(\hat{\mu_{i}}-\mu_{i})\\
 &  & +\cdots
\end{array}
$.

$\therefore$ $\sum_{i=1}^n{sim(a,i)\mu_i}\geq
\sum_{i=1}^n{sim(a,i)\hat{\mu_i}}$. According to Equation \eqref{ALPHA}, we have $\alpha \geq \hat{\alpha}$.

Therefore, the method, which selects more users from the first group, is more reliable on the predicted rating value.
\end{proof}

\begin{theorem}
\label{t1}
If $p>1-\left(\frac{n-k}{n}\right)^{\omega_1}$, the recommendation accuracy performance of Partitioned Probabilistic Neighbour Selection is better than PNCF method \cite{ZHU2014} and Probabilistic Neighbour Selection \cite{ADAMOPOULOS2014}.
\end{theorem}

\begin{proof}
We firstly demonstrate the best case for the PNCF method \cite{ZHU2014} and Probabilistic Neighbour Selection \cite{ADAMOPOULOS2014}: $sim(a,1)=\cdots=sim(a,k)=1>0=sim(a,k+1)=\cdots=sim(a,n)$.

$\therefore$ $k=k\mu_1+(n-k)\mu_n.$

According to Equation \eqref{MAIN}, $A={(1-\mu_1)^{1/{\omega_1}}=(1-\mu_n)^{1/{\omega_n}}}$, $\mu_n=1-(1-\mu_1)^{1/{\omega_1}}.$ Let $\Delta=\mu_1-\mu_n=(1-\mu_1)^{1/{\omega_1}}-(1-\mu_1)$, then $\mu_1={\frac{k}{n}+\frac{(n-k)\Delta}{n}<\frac{k}{n}+\Delta}$, namely, $\mu_1<\frac{k}{n}+(1-\mu_1)^{1/{\omega_1}}-(1-\mu_1)$, then $\mu_1<1-\left(\frac{n-k}{n}\right)^{\omega_1}.$

In PNCF method \cite{ZHU2014} and Probabilistic Neighbour Selection \cite{ADAMOPOULOS2014}, $\sum_{i\in group_1}\mu_i\leq k\mu_1$, while in Partitioned Probabilistic Neighbour Selection, $\sum_{i\in group_1}\mu_i=pk.$ Therefore, according to Lemma \ref{l2}, when $p>1-\left(\frac{n-k}{n}\right)^{\omega_1}$, $\alpha \geq \hat{\alpha}$, namely, the recommendation accuracy of Partitioned Probabilistic Neighbour Selection is better than PNCF method \cite{ZHU2014} and Probabilistic Neighbour Selection \cite{ADAMOPOULOS2014}.
\end{proof}

Since we have qualitatively analysed the recommendation accuracy performance between Partitioned Probabilistic Neighbour Selection and PNCF method \cite{ZHU2014} and Probabilistic Neighbour Selection \cite{ADAMOPOULOS2014}, now we provide the quantitative analysis of our Partitioned Probabilistic Neighbour Selection. Let $\alpha_0$ be the initial accuracy metric.

$\because \frac{\sum_{i=1}^{k}sim(a,i)\mu_{i}}{\sum_{i=1}^{k}sim(a,i)}-\frac{\sum_{i=1}^{k}\mu_{i}}{\sum_{i=1}^{k}1}=\frac{\sum_{i=1}^{k-1}\sum_{j=i+1}^{k}(sim(a,i)-sim(a,j))(\mu_{i}\mu_{j})}{k\sum_{i=1}^{k}sim(a,i)} \geq 0,$

then we have $\frac{\sum_{i=1}^k \mu_i}{\sum_{i=1}^k 1}\leq\frac{\sum_{i=1}^k sim(a,i)\mu_i}{\sum_{i=1}^k sim(a,i)}$.

$\begin{array}{lll}
\text{Namely}, & & p=\frac{pk}{k}=\frac{\sum_{i=1}^k \mu_i}{\sum_{i=1}^k 1}\leq\frac{\sum_{i=1}^k sim(a,i)\mu_i}{\sum_{i=1}^k sim(a,i)}\\
& \leq & \frac{\sum_{i=1}^k sim(a,i)\mu_i+\sum_{i=k+1}^{2k} sim(a,i)\mu_i+\cdots}{\sum_{i=1}^k sim(a,i)}\\
& = & \frac{\alpha_0}{\sum_{i=1}^k sim(a,i)}.
\end{array}$

Thus, $p\leq\frac{\alpha_0}{\sum_{i=1}^k sim(a,i)}$. Namely, when $p\geq\frac{\alpha_0}{\sum_{i=1}^k sim(a,i)}$ the actual accuracy $\alpha$ must be greater than $\alpha_0$. Therefore, we give the range of $p$'s value to guarantee the accuracy metric $\alpha\geq\alpha_0$, $p\in[\frac{\alpha_0}{\sum_{i=1}^k sim(a,i)},1]$.

\subsection{Security Analysis}
In this section, we firstly provide the range of $p$, so that our approach guarantees the system security against $k$NN attack. Next, we present the quantitative analysis by providing a relationship between the the probabilistic parameter $p$ and the security metric $\beta$.

In PNCF method \cite{ZHU2014}, according to Equation \eqref{PMF}, the probability mass function is:
\begin{equation}
PMF=\bm{I(x)}=\int_0^1\prod_{i=1}^n(1-t^{\omega_i/d})^{x_i}\text{d}t,
\end{equation}
\begin{equation}
d=\bm{\omega\cdot(m-x)}=\sum_{i=1}^n\omega_i(1-x_i),
\end{equation}
where, $\bm{x}=(x_1,x_2,\ldots,x_n)$, $\bm{\omega}=(\omega_1,\omega_2,\ldots,\omega_n)$.

For the case of selecting the top-$k$ users, we have:
\begin{equation}
x_i=
\begin{cases}
1 & i\text{\ensuremath{\in}[1,k]}\\
0 & i\in[k+1,n]
\end{cases}.
\end{equation}
Thus, the probability of selecting top-$k$ users in PNCF method \cite{ZHU2014} and Probabilistic Neighbour Selection \cite{ADAMOPOULOS2014} is:
\begin{equation}
\Pr=\bm{I(x)}=\int_0^1\prod_{i=1}^k(1-t^{\omega_i/d})\text{d}t>0,
\end{equation}
\begin{equation}
d=\bm{\omega\cdot(m-x)}=\sum_{i=k+1}^n\omega_i.
\end{equation}
In Partitioned Probabilistic Neighbour Selection, because we actually select $\lceil pk\rceil$ users from the top-$k$ users, when $p \leq \frac{k-1}{k}$, the probability of selecting top-$k$ users as the final $k$ neighbours is $0$, namely, we provide the absolute security against the $k$NN attack when setting $p \leq \frac{k-1}{k}$.

To compute the value of $\beta$ according to our selection process, we select $p(1-p)^{i-1}k$ users from group $i$, so the first time we select one user from a group, the number $j$ of this group obeys the following inequation:
\begin{equation}
\begin{array}{ll}
 & p(1-p)^{j-1}k<\frac{3}{2}\\
\Rightarrow & j>1+{\frac{(\ln{3}-\ln{2})-\ln{pk}}{\ln(1-p)}}.
\end{array}
\end{equation}
Before the group $j+1$, we have selected $pk+p(1-p)k+ \cdots +p(1-p)^{j-1}k$ users, there are $(1-p)^{j-1}k$ users can be selected. Since the each of the $(1-p)^{j-1}k$ comes from one group, the total number of the groups where the $k$ neighbours come from is:
\begin{equation}
\begin{array}{ll}
\beta & = (j-1)+{\frac{(1-p)^{j-1}k}{1}}\\
 & = (j-1)+(1-p)^{j-1}k.
\end{array}
\end{equation}

\subsection{Analysis Results}
According to the previous analysis, when setting the probabilistic parameter $p$ as $1-\left(\frac{n-k}{n}\right)^{\omega_1}<p\leq\frac{k-1}{k}$, our Partitioned Probabilistic Neighbour Selection achieve better performance of recommendation accuracy than Private Neighbour Selection \cite{ZHU2014} and Probabilistic Neighbour Selection \cite{ADAMOPOULOS2014}. Then we give the the relationship between the accuracy metric $\alpha$ and security metric $\beta$ of our Partitioned Probabilistic Neighbour Selection by the following equation:
\begin{equation}
\label{final}
\begin{cases}
p\in[\frac{\alpha_{0}}{\sum_{i=1}^{k}sim(a,i)},1]\\
j=\left\lceil 1+\frac{(\ln3-\ln2)-\ln pk}{\ln(1-p)}\right\rceil \\
\beta=(j-1)+(1-p)^{j-1}k
\end{cases}
\end{equation}
We guarantee to achieve $\alpha_0$ accuracy against $\beta$-$k$NN attack.

\subsection{A representative Example}
In this section, we show a simple but representative example of the range of the probabilistic parameter $p$. Suppose $k=\theta n$, $\theta\in(0,1]$, we know the lower bound of $p$, $1-\left(\frac{n-k}{n}\right)^{\omega_1}=1-(1-\theta)^{\omega_1}$, is a monotone-increasing function of $\theta$. Because the value of $\theta$ is always small($k\in[30,50]$ and $n$ is always greater than 1000), the value of the lower bound of $p$ will be very small. In the mean time, consider the upper bound of $p$, it would be a number close to 1. Therefore, the range of value $p$ is very large in the set of $(0,1)$.

Now we will show an example in a real scenario. Let $k=50$, $n=500$, $\epsilon=1$, $RS=1$, so the lower bound of $p$ would be
\begin{equation}
1-\left(\frac{n-k}{n}\right)^{exp(\frac{\epsilon}{4k\times RS})}=1-\left(\frac{500-50}{500}\right)^{exp(\frac{1}{4\times 50\times 1})}\approx 0.1,
\end{equation}
and the upper bound of $p$ would be $\frac{k-1}{k}=\frac{50-1}{50}=0.98$. Thus, in the above real scenario, when we set $p$ in the range of $(0.1,0.98]\subset[0,1)$, the Partitioned Probabilistic Neighbour Selection would yield better performance of recommendation accuracy against $k$NN attack.

\section{Performance Evaluation}
\label{EXP}
In Section \ref{EVA}, we theoretically analyse the performance on both recommendation accuracy and system security, and prove that to successfully preserve customer's privacy against $k$NN attack, our method ensures a better performance of recommendation accuracy than the PNCF method \cite{ZHU2014} and Probabilistic Neighbour Selection \cite{ADAMOPOULOS2014}. In this section, we compare the recommendation accuracy between Partitioned Probabilistic Neighbour Selection and global Neighbour Selection \cite{ZHU2014} and Probabilistic Neighbour Selection \cite{ADAMOPOULOS2014} by the experiments on real world dataset.

The dataset in the experiments is the MovieLens dataset\footnote{http://grouplens.org/datasets/movielens/}. The MovieLens dataset consists of 100,000 ratings (1-5 integral stars) from 943 users on 1682 films, where each user has voted more than 20 films, and each film received 20$-$250 users' rating. Specifically, we randomly select one rating of a random user, and then predict this user's potential value by $k$ Nearest Neighbour ($k$NN), Partitioned Probabilistic Neighbour Selection (PPNS), Probabilistic Neighbour Selection (nPNS) \cite{ADAMOPOULOS2014}, Private Neibgbhour Selection Collaborative Filtering (PNCF) \cite{ZHU2014}.

In this paper, we use a famous  measurement metric, Mean Absolute Error (MAE) \cite{WILLMOTT2005, ZHU2014}, to measure the recommendation accuracy:
\begin{equation}
\label{mae}
MAE=\frac{1}{T}\sum_{i\in T}|r_{ai}-\hat{r}_{ai}|,
\end{equation}
where $r_{ai}$ is the real rating of user $u_a$ on item $t_i$, and $\hat{r}_{ai}$ is the predicting rating, $T$ is the test times. To guarantee a reasonable experimental result, in our experiments, $r_{ai} \neq 0$. Clearly, a lower MAE value denotes a better prediction accuracy. Note that in each experiment, we consider the $k$NN CF recommendation method as a baseline (the best method on accuracy performance).

In our experiments, we compute the parameter $RS$ by the previous theory \cite{ZHU2014}. We set $T=10,000$, namely, we do the experiments 10,000 times to compute the MAE. Specifically, we randomly select one target user and item at each time. Our experiments are run on user-based CF (because both $k$NN attack and $\beta$-$k$NN attack are user-based attack), and we use the cosine-based metric to compute the similarity between users. Table \ref{p} and Figure \ref{p} show the relationship between accuracy performance of Partitioned Probabilistic Neighbour Selection and parameter $p$, where we set $\epsilon=1$, $k=50$, $\rho=0.5$. Table \ref{beta} and Figure \ref{beta} show the relationship between security performance of Partitioned Probabilistic Neighbour Selection (value of $\beta$) and parameter $k$, where the total partition number is 19. Table \ref{k} and Figure \ref{k} show the relationship between accuracy performance of all the four methods and parameter $k$, where we set $\epsilon=1$, $p=0.5$, $\rho=0.5$. Table \ref{rho} and Figure \ref{rho} show the relationship between accuracy performance of PNCF \cite{ZHU2014} and parameter $\rho$, where we set $\epsilon=1$, $p=0.5$, $k=50$.

\begin{figure}[!h]
\makeatletter\def\@captype{figure}\makeatother
\begin{minipage}[b]{0.5\textwidth} 
\centering 
\includegraphics[width=2.5in]{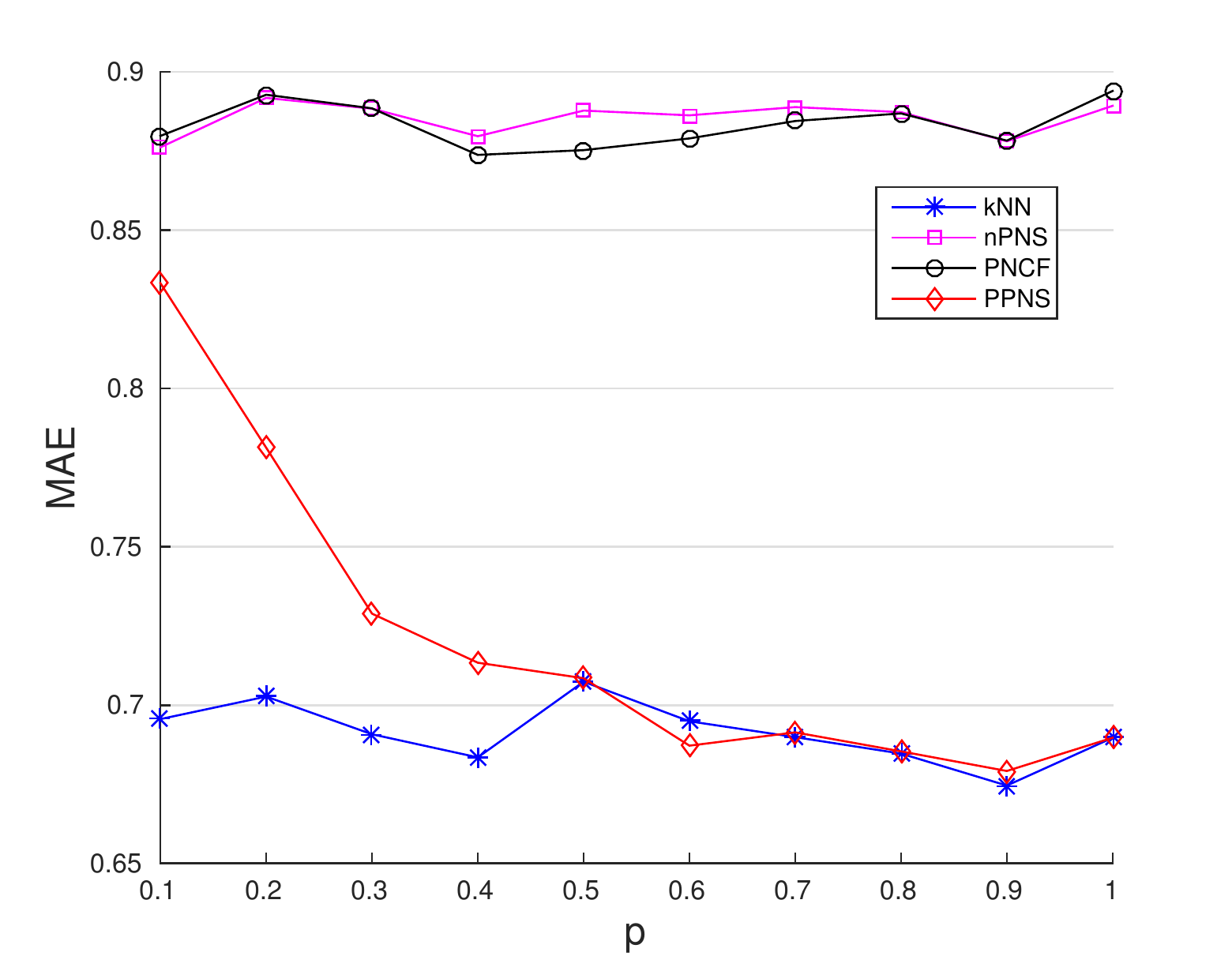}
\caption{Impacts of p on accuracy ($\epsilon=1$, $k=50$, $\rho=0.5$)}
\label{p} 
\end{minipage}
\makeatletter\def\@captype{table}\makeatother
\begin{minipage}[b]{0.5\textwidth} 
\centering
\begin{tabular}{ccccccccc}
\hline
$p$ &0.1 &0.2 &0.3 &0.4 &0.5\\
\hline
$k$NN &0.6956 &0.7027 &0.6908 &0.6835 &0.7074\\ 
PPNS &0.8333 &0.7813 &0.7289 &0.7134 &0.7085\\  
nPNS &0.8762 &0.8918 &0.8884 &0.8797 &0.8878\\
PNCF &0.8798 &0.8928 &0.8885 &0.8738 &0.8753\\ 
\hline
\hline
$p$  &0.6 &0.7 &0.8 &0.9 &1.0\\ 
\hline
$k$NN  &0.6849 &0.6899 &0.6847 &0.6746 &0.6897\\
PPNS  &0.6872 &0.6914 &0.6854 &0.6792 &0.6897\\
nPNS  &0.8863 &0.8889 &0.8873 &0.8781 &0.8893\\
PNCF  &0.8790 &0.8845 &0.8869 &0.8783 &0.8940\\
\hline
\end{tabular} 
\caption{Impacts of p on accuracy ($\epsilon=1$, $k=50$, $\rho=0.5$)} 
\label{p} 
\end{minipage} 
\end{figure}

\begin{figure}[!h]
\makeatletter\def\@captype{figure}\makeatother
\begin{minipage}[b]{0.5\textwidth} 
\centering 
\includegraphics[width=2.5in]{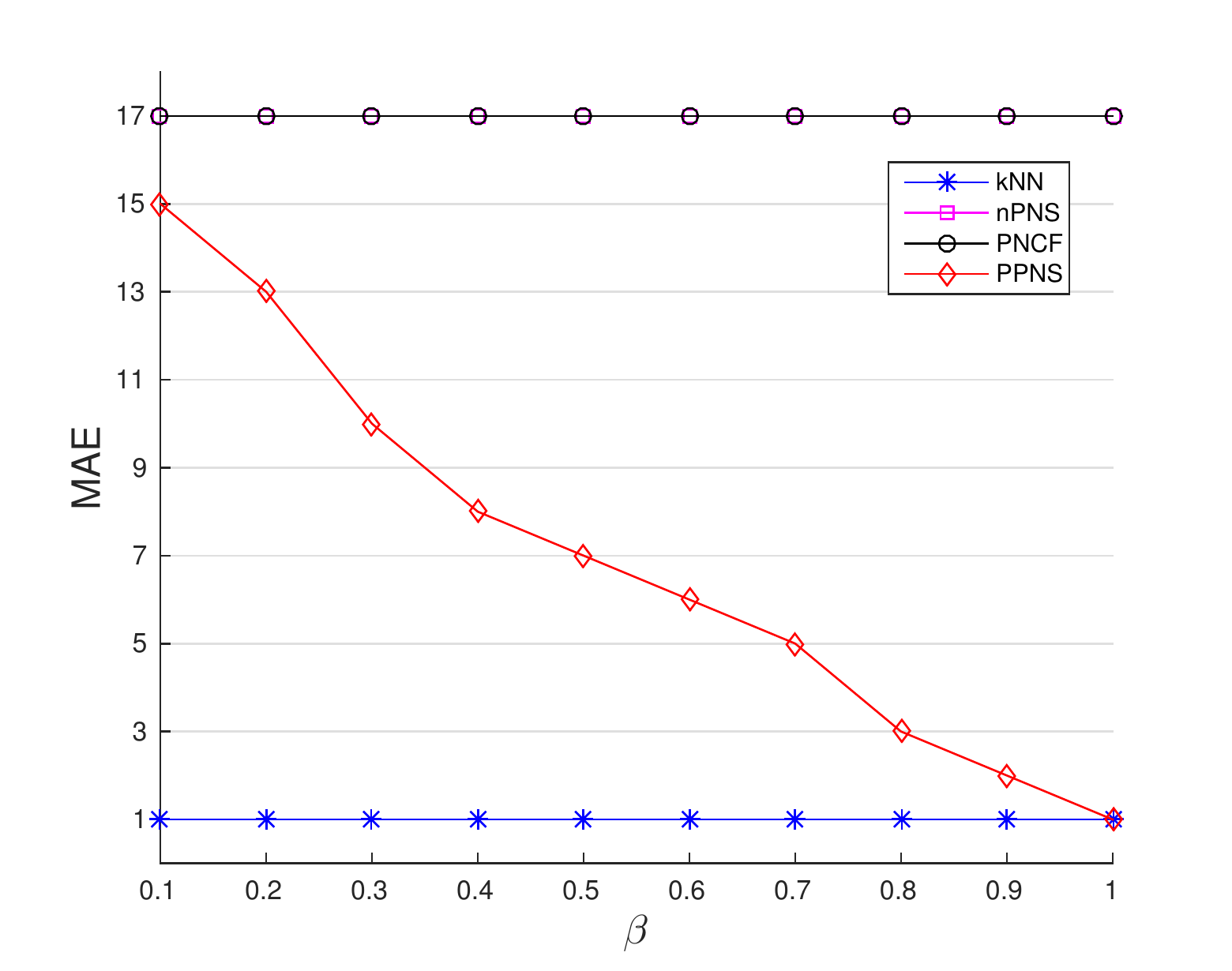}
\caption{Impacts of p on security (total partition number = 19)}
\label{beta} 
\end{minipage}
\makeatletter\def\@captype{table}\makeatother
\begin{minipage}[b]{0.5\textwidth} 
\centering
\begin{tabular}{ccccccccc}
\hline
$p$ &0.1 &0.2 &0.3 &0.4 &0.5\\
\hline
$k$NN &1 &1 &1 &1 &1\\ 
PPNS &15 &13 &10 &8 &7\\  
nPNS &17 &17 &17 &17 &17\\
PNCF &17 &17 &17 &17 &17\\ 
\hline
\hline
$p$  &0.6 &0.7 &0.8 &0.9 &1.0\\ 
\hline
$k$NN  &1 &1 &1 &1 &1\\
PPNS  &6 &5 &3 &2 &1\\
nPNS  &17 &17 &17 &17 &17\\
PNCF  &17 &17 &17 &17 &17\\
\hline
\end{tabular} 
\caption{Impacts of $p$ on security (total partition number = 19)}
\label{beta} 
\end{minipage} 
\end{figure}

\begin{figure}[!h]
\makeatletter\def\@captype{figure}\makeatother
\begin{minipage}[b]{0.5\textwidth} 
\centering 
\includegraphics[width=2.5in]{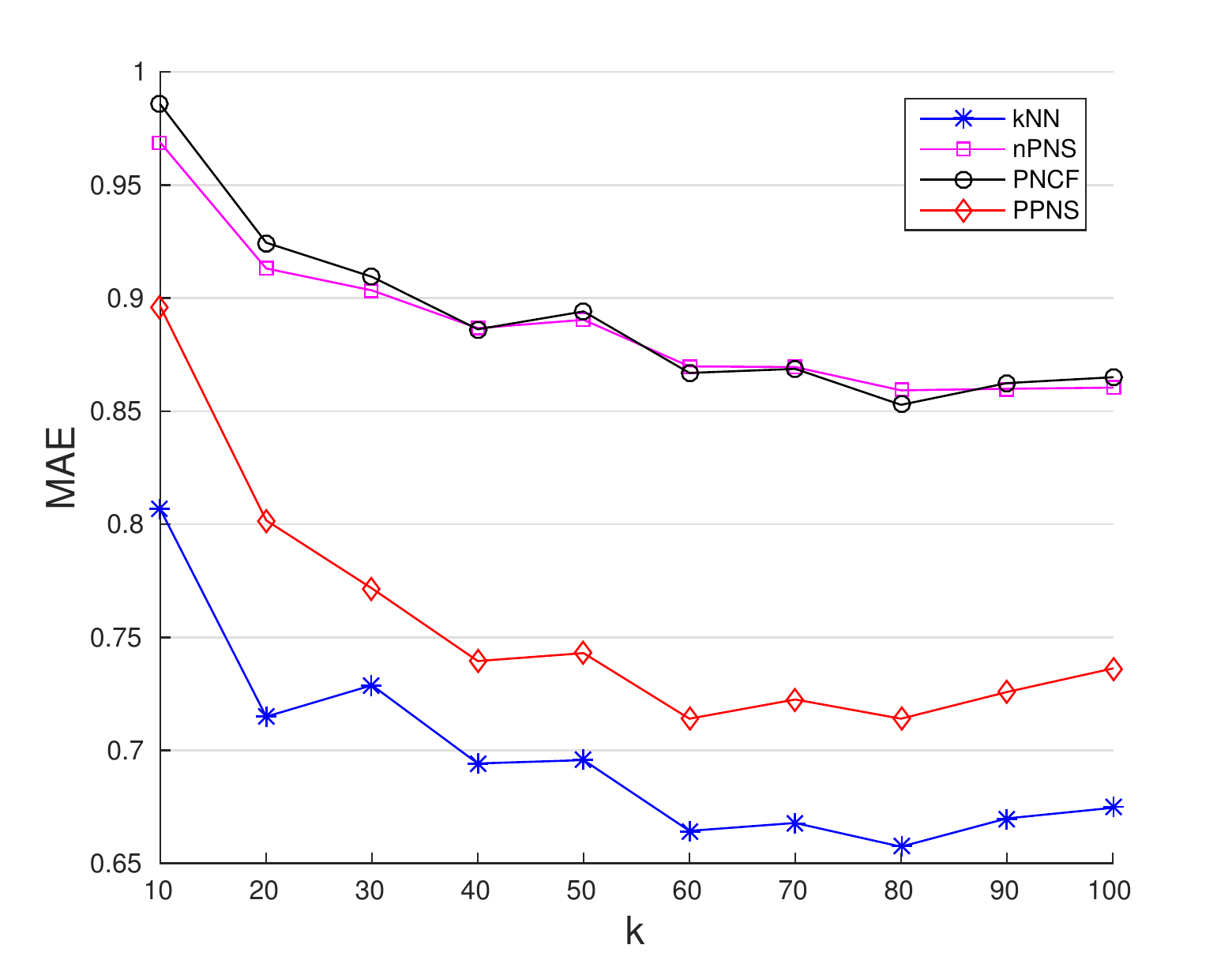}
\caption{Impacts of k on accuracy ($\epsilon=1$, $p=0.5$, $\rho=0.5$)}
\label{k} 
\end{minipage}
\makeatletter\def\@captype{table}\makeatother
\begin{minipage}[b]{0.5\textwidth} 
\centering
\begin{tabular}{ccccccccc}
\hline
$k$ &10 &20 &30 &40 &50\\
\hline
$k$NN &0.8065 &0.7149 &0.7288 &0.6942 &0.6957\\ 
PPNS &0.8962 &0.8017 &0.7716 &0.7395 &0.7430\\  
nPNS &0.9687 &0.9131 &0.9034 &0.8867 &0.8904\\
PNCF &0.9856 &0.9245 &0.9094 &0.8862 &0.8941\\ 
\hline
\hline
$k$  &60 &70 &80 &90 &100\\ 
\hline
$k$NN  &0.6644 &0.6679 &0.6574 &0.6699 &0.6746\\
PPNS  &0.7140 &0.7225 &0.7140 &0.7258 &0.7362\\
nPNS  &0.8698 &0.8695 &0.8592 &0.8599 &0.8604\\
PNCF  &0.8669 &0.8687 &0.8528 &0.8624 &0.8650\\
\hline
\end{tabular} 
\caption{Impacts of k on accuracy ($\epsilon=1$, $p=0.5$, $\rho=0.5$)}
\label{k}
\end{minipage} 
\end{figure}

\begin{figure}[!h]
\makeatletter\def\@captype{figure}\makeatother
\begin{minipage}[b]{0.5\textwidth} 
\centering 
\includegraphics[width=2.5in]{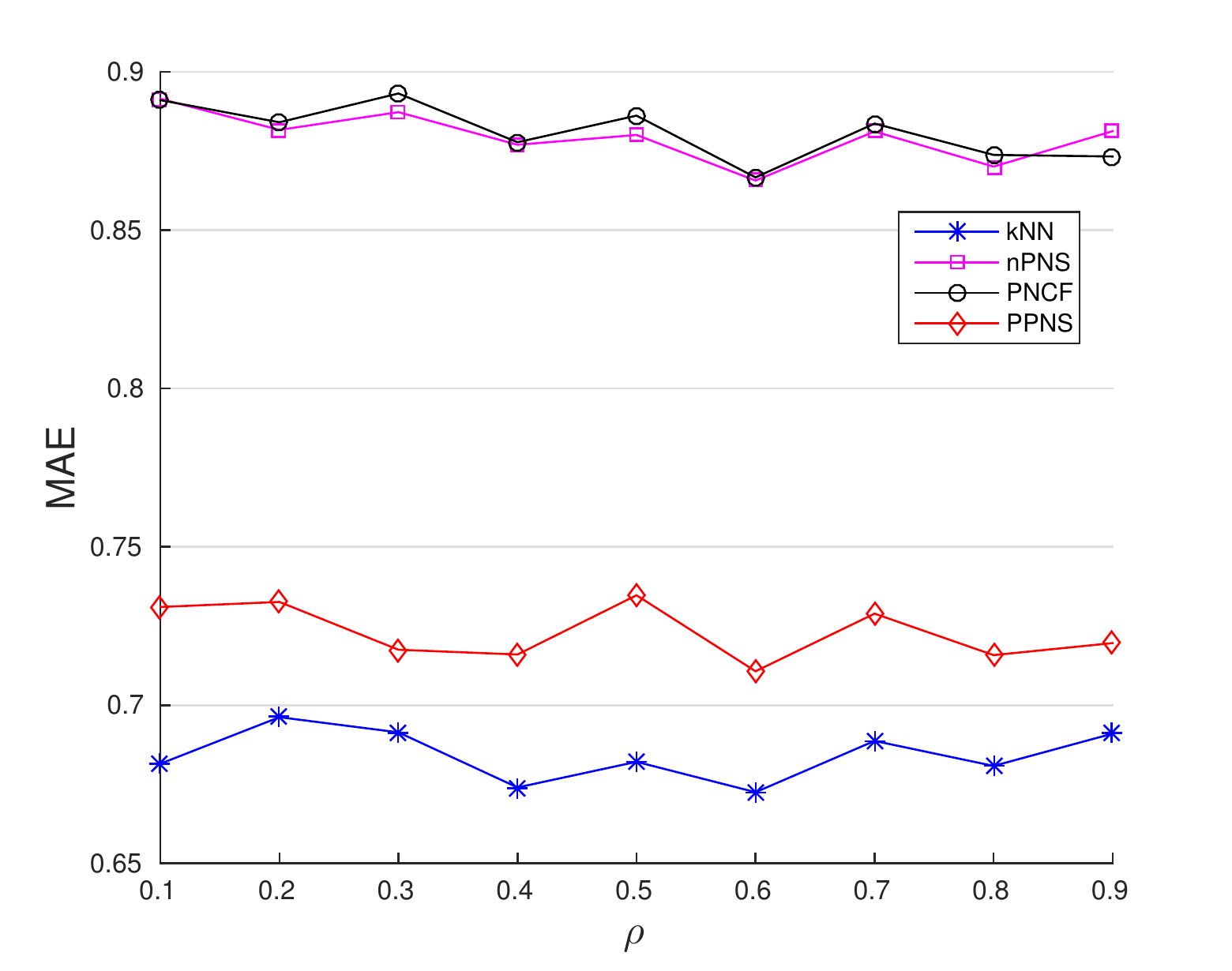}
\caption{Impacts of $\rho$ on accuracy ($\epsilon=1$, $p=0.5$, $k=50$)}
\label{rho}
\end{minipage}
\makeatletter\def\@captype{table}\makeatother
\begin{minipage}[b]{0.5\textwidth} 
\centering
\begin{tabular}{ccccccccc}
\hline
$\rho$ &0.1 &0.2 &0.3 &0.4 &0.5\\
\hline
$k$NN &0.6815 &0.6962 &0.6914 &0.6740 &0.6821\\ 
PPNS &0.7310 &0.7326 &0.7175 &0.7160 &0.7374\\  
nPNS &0.8915 &0.8817 &0.8873 &0.8770 &0.8801\\
PNCF &0.8911 &0.8841 &0.8932 &0.8778 &0.8862\\ 
\hline
\hline
$\rho$  &0.6 &0.7 &0.8 &0.9\\ 
\hline
$k$NN  &0.6725 &0.6887 &0.6808 &0.6910\\
PPNS  &0.7107 &0.7289 &0.7185 &0.7196\\
nPNS  &0.8657 &0.8812 &0.8701 &0.8813\\
PNCF  &0.8667 &0.8837 &0.8738 &0.8733\\
\hline
\end{tabular} 
\caption{Impacts of $\rho$ on accuracy ($\epsilon=1$, $p=0.5$, $k=50$)}
\label{rho}
\end{minipage} 
\end{figure}

\newpage
According to the experiments results, we are clear that:
\begin{itemize}
\item[(1)] From Fig. \ref{p}, when setting $p > {1-\left(\frac{n-k}{n}\right)^{\omega_1}}$, the accuracy performance of Partitioned Probabilistic Neighbour Selection is always better than the PNCF method \cite{ZHU2014} and Probabilistic Neighbour Selection \cite{ADAMOPOULOS2014}. When the value of $p$ is close to 1, the performance of Partitioned Probabilistic Neighbour Selection is close to $k$NN method. Particularly, when $p=1$, the Partitioned Probabilistic Neighbour Selection method is the same as $k$NN method.
\item[(2)] From Fig. \ref{p} and Fig. \ref{beta}, the accuracy performance of the neighbourhood-based CF methods largely depends on the quality of neighbours. Our Partitioned Probabilistic Neighbour Selection method yields a better trade-off between the recommendation accuracy and security, as when we offer a better accuracy performance, we do not lose the security much.
\item[(3)] From Fig. \ref{k}, the size of neighbour set impacts the accuracy performance of all of the neighbourhood-based CF recommendation approaches. A large value of neighbour set size $k$ yields a better accuracy performance.
\item[(4)] From Fig. \ref{rho}, the value of $\lambda$ in \cite{ZHU2014} is not achievable because the value of $\rho$ does not impact the accuracy performance of PNCF \cite{ZHU2014}.
\end{itemize}

\section{Conclusion}
\label{CON}
Recommender systems play an important role in e-commerce. To protect users' private information during the process of filtering, the existing privacy preserving neighbourhood-bases CF methods fail to protect users' privacy in rating prediction. The global probabilistic neighbour selection methods, such as the PNCF method \cite{ZHU2014} and Probabilistic Neighbour Selection \cite{ADAMOPOULOS2014} though can protect users' privacy against $k$NN attack successfully, but provide no data utility guarantee. To overcome the weaknesses of the current methods, we propose a novel privacy preserving neighbourhood-based CF method, Partitioned Probabilistic Neighbour Selection, to ensure a required recommendation accuracy while maintaining high system security against $\beta$-$k$NN attack (generalisation of $k$NN attack). Theoretical and experimental analysis show that to provide an accuracy-assured recommendation against the most popular attack, $k$NN attack, our Partitioned Probabilistic Neighbour Selection method yields a better trade-off between the recommendation accuracy and system security than the PNCF methods \cite{ZHU2014} and Probabilistic Neighbour Selection \cite{ADAMOPOULOS2014}.

\bibliographystyle{plaint}
\bibliography{zhigang}

\end{document}